\documentclass{siamonline171218}
\PassOptionsToPackage{noabbrev}{cleveref}

\usepackage{amsmath, amssymb, amsfonts}
\usepackage{graphicx}
\usepackage{subfig}
\usepackage{float}
\usepackage{algorithm}
\usepackage{subcaption}  
\usepackage{algpseudocode}
\usepackage{fancyhdr}
\usepackage{mathrsfs}
\usepackage{tcolorbox}

\title{Effective Resistance-Based Graph Sparsification and Community Detection\
  \thanks{Submitted to the editors DATE.
    \funding{Research of JP is supported by JRF from Department of Science and Technology(DST), Government of India. Research of PB is funded by JRF from the Ministry of Human Resource Development (MHRD), Government of India.}}}

\author{Jayanta Pari, Pratibha Bhandari, Soumyendu Raha%
\thanks{Department of Computational and Data Science, Indian Institute of Science (IISc), Bangalore, India (jayantapari@iisc.ac.in, pratibhab@iisc.ac.in, raha@iisc.ac.in)}}

\usepackage{hyperref}
\begin{document}

\maketitle

\begin{abstract}
Community detection is a key task in network analysis, providing insight into the structural organization of complex systems. Effective resistance, a graph-theoretic metric derived from electrical network theory, has emerged as a powerful tool for evaluating connectivity and influence within networks. This paper proposes an effective resistance-based community detection algorithm that calculates the similarity between nodes using effective resistance values and produces a weighted graph. The sparse graph used in the algorithm is generated after computing the minimum spanning tree (MST) of the weighted graph and adopting a threshold sparsification strategy on non-MST edges. A maximum modularity approach is adopted using the Clauset-Newman-Moore algorithm on the resultant sparse graph. This algorithm is evaluated for both synthetic and real-world networks, demonstrating its effectiveness compared to popular existing methods. The result shows that the effective resistance-based approach accurately captures the structures of the community while maintaining computational efficiency.


\end{abstract}


\maketitle
\section{Introduction}
The search for communities in complex networks is a crucial problem that exists in a variety of different disciplines, such as sociology \cite{zachary1977information},\cite{girvan2002community}, biology \cite{calderer2021community},\cite{rahiminejad2019topological}, telecommunications\cite{wang2017model},\cite{ye2009telecomvis}, and physics. The concept of communities refers to subgroups of nodes that are densely connected within themselves but have only sparse connections to other groups. Communities are present in complex networks due to their modular structure. In practice, finding some modularity in a real-world network makes the job of analyzing large-scale networks easier by clustering nodes into meaningful subgroups. However, detecting communities and defining these groups is still an ill-posed problem, which is further complicated by the challenge of defining a general definition of a community, along with its persistent issue of its associated computational complexity. 

Many real-world networks typically have high edges density, and in many cases these edges carry weights to represent the connection\cite{barrat2004architecture}. Due to their dense nature, the analysis of these types of networks is computationally heavy. So some preprocessing  of data is required to solve the problem.

Graph sparsification \cite{tumminello2005tool},\cite{spielman2008graph} is a powerful approach to reduce the number of edges of a graph by maintaining certain properties of a graph. In \cite{spielman2008graph}, they proposed one sparsification algorithm based on effective resistance.

The effective resistance between two nodes of a graph has an intrinsic relation to uniform spanning trees. In particular, for any edge $e={(u,v)}$ of an undirected graph $G$, the effective resistance $R_{eff}(u,v)$ is equivalent to the probability that the edge $e$ belongs to a random spanning tree\cite{doyle1984random} of $G$. In \cite{chandra1989electrical}, it was also demonstrated that effective resistance is proportional to the commute time between the end points of edges. Furthermore,  \cite{klein2022real} showed that the square root of effective resistance of an edge is a Euclidean metric, measuring the distance between two nodes.

In this paper, we propose a new algorithm for community detection based on effective resistance. The algorithm is inspired by \cite{calderer2021community}. Here we use the Clauset-Newman-Moore algorithm \cite{clauset2004finding} for maximization of modularity and community detection. Let $G=(V,E)$ be the network, where $V$ is the vertex set and $E$ are edges. $G_{eff}$ is obtained from $G$ by evaluating the effective resistance of all edges. $G_{eff}$ has the same set of vertex, only difference is the edge weight, which are assigned based on effective resistance. However the adjacency matrix $A$ of $G_{eff}$ and $G$ are same. So to reduce the number of edges it is essential to employ a sparsification technique. By doing so, the resulting sparse network provides a clearer representation of the original network. 

Our proposed community detection algorithm leverages effective resistance for both graph sparsification and community identification.
The main contributions of our work are:
\begin{itemize}
\item
A novel sparsification technique using spanning tree that preserves community structure while reducing computational complexity.
\item
Extensive experimental evaluation on both synthetic and real-world networks demonstrating the superiority of our approach compared to state-of-the-art methods
\end{itemize}
The rest of this paper is structured as follows: Section 2 provides a review of  related work on community detection and resistance-based metrics. Section 3 outlines the theoretical foundation of effective resistance and introduces our community detection approach. Section 4 details the methodology, while Section 5 explains the experimental setup. Section 6 presents the results, Section 7 offers the comparison with existing methods. Finally, Section 8 concludes the paper with a summary of findings and potential directions for future research. 
\section{Related Work}
This section summarizes major advances in community detection algorithms and metrics based on resistance for network analysis.
\subsection{Community Detection Approaches}
The classical approaches to community discovery fall into many broad classes that consist of modularity optimization, spectral methods, statistical inference, dynamics-based techniques\cite{fortunato2010community}, and Ricci curvature-based method\cite{ni2019community}.
Modularity, introduced by Newman and Girvan (2004), is among the most widely used quality functions for evaluating community structures. Modularity is defined as the proportion of edges within communities against what is found in an analogous random network with the same degree distribution. Various algorithms have been proposed to maximize the modularity, such as the fast greedy algorithm \cite{clauset2004finding}, the Louvain algorithm\cite{blondel2008fast}, and several evolutionary strategies\cite{pizzuti2017evolutionary}.

Spectral techniques apply the eigenvectors of matrices associated with a graph, for example, the Laplacian matrix, to divide network into communities\cite{von2007tutorial}. They are well-based mathematically and are able to effectively find global community patterns.

Statistical inference methods frame community detection as the problem of fitting generative models to empirical network data. Prominent examples of this strategy include stochastic block models\cite{karrer2011stochastic} and their numerous extensions.

Dynamic-based methods mirror phenomena like stochastic movement or information diffusion in networks to outline communities. Some of the most commonly known ones are Infomap\cite{rosvall2008maps}, which is based on the minimum description length paradigm to characterize random walks, and Label Propagation\cite{raghavan2007near}, based on label diffusion through the network.

Recently, the Ricci flow-based method \cite{ni2019community} was introduced as a classical geometric approach. It is used to decompose the smooth manifold. In \cite{ni2019community}, they considered the network as a geometric object and the communities in a network as a geometric decomposition to break down the communities within the network.

Subsequently, Pizzuti and Socievole in \cite{pizzuti2021effective}, introduced a genetic algorithm for community detection that uses effective resistance to compute of node similarity. The method exhibits promising performance on both synthetic and real network data.

The following section reviews the mathematical preliminaries of the effective resistance metric and explains how it is utilized to identify meaningful community structures.

\section{Mathematical preliminary of Effective Resistance}
Here we will discuss some mathematical preliminaries of effective resistance and followed by a detailed explanation of the proposed approach.
\subsection{Effective Resistance and its Properties}
The effective resistance $R(u,v)$ between two nodes $u,v \in V$ in a graph is the voltage difference between the nodes when a current of one ampere is injected into the source node $u$ and extracted from the sink node $v$. 

An important characteristic of the effective resistance \( R_{uv} \) between nodes \( u \) and \( v \) is that it is \( \textit{bounded above} \) by the shortest-path distance in a graph~\cite{van2023graph}. This makes it different from shortest-path measure, as it takes into account the cumulative contributions of all possible paths between the nodes. Consequently, effective resistance provides a more comprehensive characterization of connectivity that captures all the topological structure of the network.
Moreover, effective resistance is strongly correlated with the \textit{commute-time distance} \( T_{uv} \), which is the expected number of steps of a random walk from node \( u \) to node \( v \) and back to \( u \). This can be written as:

\begin{equation}
T_{uv} = \mathbf{1}^T\mathbf{A} \mathbf{1} \cdot R_{uv}, 
\end{equation}

where \( \mathbf{1} \) is the all-ones vector, \( \mathbf{A} \) is the (weighted) adjacency matrix of the graph. The scalar \( \mathbf{1}^T \mathbf{A} \mathbf{1} \) is twice the total sum of edge weights of the graph. This expression points out that commute-time distance grows linearly with the effective resistance between nodes~\cite{chandra1989electrical}.\\
For a graph $G(V,E)$ we know that the relationship between node voltage and edge current is 
\begin{equation}
    L\mathcal{V}=I_{uv}(e_{u}-e_{v})
\end{equation}
The equation came from Kirchhoff's current law. Here $L$ is the Laplacian matrix of the graph and $\mathcal{V}$ is the node voltage and $I_{uv}$ is the edge current and for vertex $u\in V$, the vector $e_{u}$ is the standard basis vector of $R^n$ with value one on the
position associated to vertex $u$. From the above equation we will get $\mathcal{V}_{u}$ and $\mathcal{V}_{v}$ which will give the effective resistance $R_{uv}=\frac{\mathcal{V}_{u}-\mathcal{V}_{v}}{I_{uv}}$.\\
\begin{theorem}
For a graph $G(V,E)$ with edge weights $w_{uv}$ for edge $(u, v)\in E$, define the resistances on each edge $(u, v)\in E$ as $r_{uv} = {w}_{uv}^{-1}$. Then the effective resistance between the vertices $u$ and $v$ is\cite{vos2016methods}
\begin{equation*}
    R_{uv} = (e_{u} - e_{v})^T L^+ (e_{u} - e_{v}) = L^+_{uu}-2L^+_{uv}-L^+_{vv}.
\end{equation*}
\end{theorem}
\textbf{Proof.} 
We know that $L\mathcal{V}=I_{uv}(e_{u}-e_{v})$ from this one can obtain the expression of $\mathcal{V}$ as 
\begin{equation*}
    \mathcal{V}=I_{uv}L^+(e_{u}-e_{v})
\end{equation*}
Here $L^+$ is the pseudoinverse of $L$
We have expression of $R_{uv}=\frac{\mathcal{V}_{u}-\mathcal{V}_{v}}{I_{uv}}=(e_{u} -e_{v})^T \frac{\mathcal{V}}{I_{uv}}$
putting the value of $\frac{\mathcal{V}}{I_{uv}}$ we will get the desired result.
The equivalent equation will come after doing the eigendecomposition of $L^+$,
$L^+ = VD^+V^T$, where $V$ is the eigenvector, now 
\begin{equation*}
    R_{uv}= (e_{u} - e_{v})^T VD^+V^T(e_{u} - e_{v})
\end{equation*}
this will give the desired expression.\\
\begin{theorem}
The eﬀective resistance matrix $R$ is a Euclidean distance matrix\cite{arpita2008}.
\end{theorem}
\textbf{Proof.}
A symmetric matrix $S\in R^{n\times n}$ is called a Euclidean distance matrix if there exists a set of $n$ vectors $p_{1},p_{2}...p_{n}\in R^n$ such that $S_{ij}={\parallel p_{i} - p_{j} \parallel}^2$. One result of Euclidean distance matrix is that a matrix $S$ is a euclidean distance matrix if and only if it has nonnegative entries, zero diagonal and is negative semidefinite on $1^T$ i.e $p^TSp \leq 0$ with the condition $1^Tp=0$\cite{gower1985properties}.
Now as it is known that effective resistance values are always nonnegative and also the effective resistance between the same vertex is $0$, so the first two properties of Euclidean distance is satisfied for $R$ matrix.\\
let us take one vector $p$ such that $1^Tp=0$, now 
\begin{equation*}
    p^TRp=p^T(1diag(L^+)^T + diag(L^+)1^T -2L^+)p
\end{equation*}
\begin{equation*}
    =-2p^TL^+p\leq 0
\end{equation*}
So resistance matrix satisfies all the properties of Euclidean distance matrix. Here $diag(L^+)$ is the vector consisting of the diagonal elements of $L^+$.\\
One thing can be shown that we can construct a set of vectors $p_{1},p_{2}...p_{n}\in R^n$ such that  $R_{ij}={\parallel p_{i} - p_{j} \parallel}^2$\\
Let $p_{i}$ be the columns of $(L^+)^\frac{1}{2}$. The Gram matrix constructed from these vectors corresponds to $L^+$
\begin{equation*}
    p_{i}^Tp_{j}=L^+_{ij}
\end{equation*}
Now Expanding ${\parallel p_{i} - p_{j} \parallel}^2$ we will get 
\begin{equation*}
    {\parallel p_{i} - p_{j} \parallel}^2=p_{i}^Tp_{i} + p_{j}^Tp_{j}-2p_{i}^Tp_{j}
\end{equation*}
\begin{equation*}
    =L^+_{ii} +L^+_{jj} -2L^+_{ij}
\end{equation*}
which is $R_{ij}$.\\
Hence, the resistance matrix $R$ is an Euclidean distance matrix, so the square root of its entry will be the Euclidean distance, and it will also satisfy the triangle inequality.
\begin{equation*}
    R^\frac{1}{2}_{ij}\leq  R^\frac{1}{2}_{ik} +  R^\frac{1}{2}_{kj}
\end{equation*}.
\begin{theorem}
For a graph $G$, the total effective resistance\cite{arpita2008} is given by
\begin{equation*}
    R_{total}=n\sum_{1}^{n}\frac{1}{\lambda_{i}}
\end{equation*}
where ($\lambda_{i}\neq 0$) are eigenvalues of the graph Laplacian matrix $L$.
\end{theorem}
\textbf{Proof.}
We know that for a graph $G$, the total effective resistance is 
\begin{equation*}
    R_{total}=\sum_{p=1}^{n-1}\sum_{q=p+1}^{n} R_{pq} = \frac{1}{2}\sum_{p=1}^{n}\sum_{q=1}^{n} R_{pq},
\end{equation*}
putting the value of $R_{pq}$ in above equation we get the expression
\begin{equation*}
    R_{total}=\frac{1}{2}\sum_{p=1}^{n}\sum_{q=1}^{n} (L^+_{pp}-2L^+_{pq} + L^+_{qq}), 
\end{equation*}
which can be written as 
\begin{equation*}
    R_{total}= n\sum_{i=1}^{n} L^+_{pp} - 1^TL^+1.
\end{equation*}
Suppose \{$l_{1}...l_{n}$\} are the orthonormal basis of eigenvectors of the graph Laplacian matrix $L$. Then the pseudo-inverse can be defined as
\begin{equation*}
    L^+l_{1}=0 ,\quad
    L^+l_{i}= \frac{1}{\lambda_{i}}l_{i},\quad \text{for}\quad i\neq 1.
\end{equation*}
It can be interpreted as $L^+1=0$ and for all $x\perp 1,\quad L^+x=y$ with $y\perp 1$
Now the expression for total resistance will be 
\begin{equation*}
    R_{total}= n Tr(L^+).
\end{equation*}
$L^+ = VD^+V^T$ can be written after eigen decomposition. We also know that a similar matrix have the same trace value, so
\begin{equation*}
    R_{total}=nTr(L^+)=nTr(D^+)=n\sum_{i=1}^{n}\frac{1}{\lambda_{i}}.
\end{equation*}
\\
There is a relationship between the determinants of submatrix of the Laplacian matrix $L$ and the effective resistance $R$.\\
\begin{theorem}
Let $G(V,E)$ be a connected, weighted graph with  $n\geq 3$ vertices, where each edge $(u,v)\in E$ has a positive weight $w_{uv}\in \mathbb{R}^+$. The resistance of an edge $e=(u,v)\in E$ is defined as $r_{e}=\frac{1}{w_{uv}}$. For any pair of vertices $u,v\in V$, the effective resistance between them is given by\cite{vos2016methods}
\begin{equation*}
    R_{uv}=\frac{det(L(u,v))}{det(L(u))}.
\end{equation*}
Here $L$ represents the graph Laplacian, and $L(u,v)$ and $L(u)$ denote the Laplacian matrix with rows and columns corresponding to vertices $u$ and $v$.
\end{theorem}
As from the above theorem, we can see certain relationship between resistance distance and determinant of Laplacian. So there must be a relationship between spanning tree and resistance distance. For an unweighted graph $G(V,E)$, the number of spanning trees is given by $det(L(u))$ for any $u\in V$, as stated by the Matrix-Tree Theorem. Moreover, it is shown that $det(L(u,v)$ corresponds to the number of spanning trees in $G$ including the edge $(u,v)$.
\begin{theorem}
For a graph $G(V,E)$, the number of spanning tree in $G$ is equal to $det(L(u))$ for any $u\in V$ \cite{vos2016methods},
\begin{equation*}
    T(G)= \det (L(u)).
\end{equation*}
\end{theorem}
This is the \textbf{Matrix-Tree theorem}. Another theorem is 
\begin{theorem}For a weighted graph $G(V, E)$ with Laplacian $L$, the number of spanning trees in the graph resulting from contracting the edge between vertices $u$ and $v$ is same as the determinant of the modified Laplacian matrix, $\text{det}(L(u, v))$\cite{vos2016methods}. It means
\begin{equation*}
    T(G/uv)=\det(L(u,v)).
\end{equation*}
\end{theorem}

The resultant theorem by the combination of the theorem $4$, $5$ and $6$ is given as,
\begin{theorem}
Let $G(V,E)$ be the graph and $u,v\in V$, the effective resistace of the edge $(u,v)$ is\cite{vos2016methods} 
\begin{equation*}
    R_{uv}= \frac{T(G/uv)}{T(G)}
\end{equation*}
\end{theorem}
As per Interlacing Theorem\cite{van2023graph,so1999rank}, the effective resistance between any two nodes of a network decreases or remains the same with a rise in the weights of the edges. It is known as Rayleigh's monotonicity law.
\begin{theorem}
The effective resistance of a graph decreases strictly when new edges are added or when the weights of existing edges are increased.
\end{theorem}
Next we will discuss about the methodology\\
\section{Methodology}
Here we will discuss our method \textbf{ERSCD}(effective resistance based sparsification and community detection).
The algorithm is partitioned into $3$ parts. First part involves efficient computation of the effective resistance of the network, and the second part focuses to sparsify the graph based on similarity matrix and the last part performs community detection.\\
\paragraph{\textbf{Efficient effective resistance calculation:}}
In Theorem $1$ we have shown one compact formula for effective resistance calculation. Here we have to calculate the pseudoinverse of the laplacian matrix.
Also from equation $2$ we can see a system of linear equation for unit current is,
\begin{equation}
    L\mathcal{V} = (e_{u} - e_{v}).
\end{equation}
The matrix $L$ is a sparse symmetric positive semidefinite matrix. 
To solve this sparse system we use the Algebraic Multigrid (AMG) solver\cite{ruge1987algebraic}, an efficient iterative solver to solve large sparse linear system. For graphs with $n$ nodes and $m$ edges, AMG solves Laplacian systems in nearly linear time $\mathcal{O}(n+m)$, so if we solve this for edge the total run time will be $m \mathcal{O}(m+n)$.

But the main issue is that for calculating the effective resistance of each pair the linear solver will take $\mathcal{O}(m+n)$ time. But in case of theorem $1$ the pseudoinverse calculation is one time and then to calculate the effective resistance of each pair it will take $\mathcal{O}(1)$ time, that's why the theorem $1$ is suitable for exact effective resistance calculation. 

We have shown the time complexity of different methods on SBM graph with fixed mixing parameter $\mu=0.2$ and number of nodes ranging from $100$ to $900$ in the graphs below.
\begin{figure}
    \centering
    \includegraphics[width=0.6\textwidth]{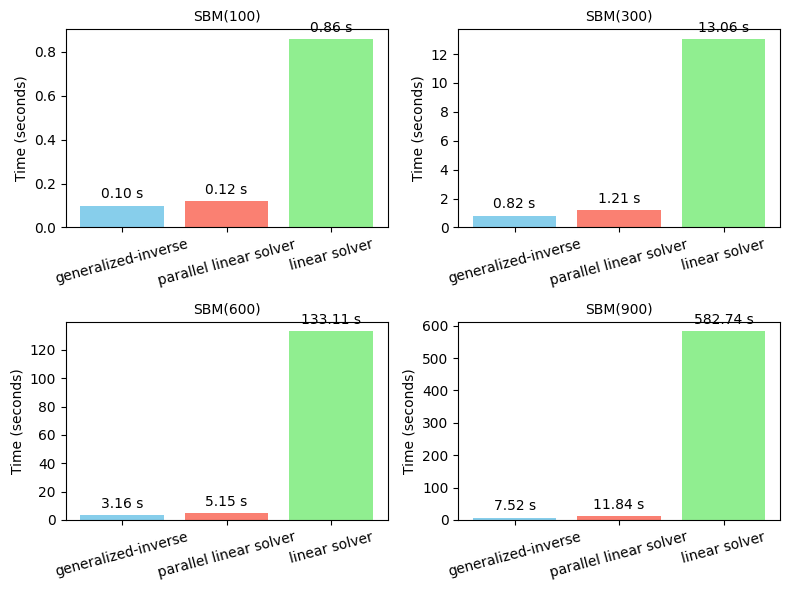}
    \caption{Time taken by different methods to compute pseudo inverse}
\end{figure}
Otherwise to compute the approximate effective resistance in $\mathcal{O}(m \log n)$ time one can use Spielman and Srivastava\cite{spielman2008graph} algorithm. Here the effective resistance between two node $u,v$ is 
\begin{equation*}
    R(u,v)=\parallel w^\frac{1}{2}ML^+(e_{u}-e_{v}) \parallel^2_{2}.
\end{equation*}
Here the incidence matrix is given by $M$ and $w$ is edge weight. The pairwise distance between vectors $(w^\frac{1}{2}ML^+e_{u})u\in V$ is the effective resistance. They have used Johnson-Lindenstrauss Lemma and project the vectors onto an $\mathcal{O}(\log n)$-dimensional random subspace. We have tested this approach to compute effective resistance but we did not get satisfactory result.\\
\paragraph{\textbf{Graph sparsification:}}
After computing the effective resistance of the network, we calculate the Gaussian function of the resistance distance and assign it as the weight of edges  of the network.
\begin{equation*}
    w_{uv}=exp(\frac{-R_{uv}^2}{2\sigma})
\end{equation*}
Here $\sigma$ is a parameter (considered 1 for our experiments). As edge weights are changed, the effective resistance of the network will also change (theorem $8$). So we need to calculate the effective resistance of the modified graph, which is $R^{modified}$. 

Also from theorem $6$, we know that the square root of effective resistance defines Euclidean distance. To make $R^{modified}$, the Euclidean distance we took the Hadamard product with adjacency matrix and then took its square root
\begin{equation*}
    R^{modified}_{H}=\sqrt{R^{modified}\circ A}.
\end{equation*}
Next, we normalize the values of the effective resistance and make it symmetric,
\begin{equation*}
    R^{modified}_{H}=\frac{1}{2}(\frac{R^{modified}_{H}}{max(R^{modified}_{H})}+\frac{(R^{modified}_{H})^T}{max(R^{modified}_{H})}).
\end{equation*}
To create the similarity matrix, we will use the following normalized effective resistance,
\begin{equation*}
    S=\mathcal{J}-R^{modified}_{H}
\end{equation*}
where $\mathcal{J}$ is the matrix with all entries $1$.

Next, minimum spanning tree is computed using prims algorithm with effective resistance as weight and separate the MST and non-MST edges.

Next, sort the non-MST edges by the ascending order of similarity matrix, because to remove the non similar edges which are not necessary.

Next, we give $pe$, percentage of edges to delete from the non-MST edges, and after deleting the edges, the sparsified graph is built with weight as similarity values.\\
\paragraph{\textbf{Community Detection:}}
Now to detect community of this sparsified graph, we use the Clauset algorithm\cite{clauset2004finding}. It is a greedy optimization technique. The algorithm runs in $\mathcal{O}\left(m d \log n\right)$, where $m$ is the number of edges and $n$ is the number of nodes and $d$ is the depth of
the “dendrogram” describing the network’s community. For sparse networks, $m\approx n$ and $d\approx \log(n)$. For such network the algorithm run time is $\mathcal{O}\left(n log^2(n) \right)$. Before this algorithm, greedy based algorithm \cite{newman2004finding} with the run time of $\mathcal{O}(n^2)$ has been proposed. But in the algorithm\cite{clauset2004finding} Clauset et.al. have used more sophisticated data structures to make it faster. 
Before describing the algorithm, we need to describe about modularity on which the algorithm is based.
\paragraph{\textbf{Modularity}}
Modularity\cite{newman2004finding} measures clustering quality by comparing the density of intra-community edge to that of a random graph. It takes values ranging from \([-1, 1)\), with higher values indicating stronger community structure. Modularity does not require ground truth labels and is computationally efficient for large networks. However, it may suffer from \textit{resolution limit}\cite{sun2014analysis}.

For a network \( G \) with \( n \) nodes, \( m \) edges, and adjacency matrix \( A \), modularity \( \mathcal{Q} \) is defined as:

\begin{equation}
\mathcal{Q} = \frac{1}{2m} \sum_{u,v} \left( A_{uv} - \frac{d_u d_v}{2m} \right) \delta(c_u, c_v)
= \sum_{i=1}^{c} \left( e_{ii} - a_i^2 \right)
\end{equation}

where:
\begin{itemize}
    \item \( d_x \): degree of node \( x \),
    \item \( \delta(c_u, c_v) = 1 \): if nodes \( u \) and \( v \) are in the same community, 0 otherwise,
    \item \( e_{ij} \): fraction of edges between communities \( i \) and \( j \),
    \[
    e_{ij} = \frac{1}{2m} \sum_{u \in C_i} \sum_{v \in C_j} A_{uv}
    \]
    \item \( a_i \): fraction of edges connected to nodes in community \( i \), given by:
    \[
    a_i = \frac{d_i}{2m}
    \]
    \item 
    $m=\frac{1}{2}\sum_{uv}A_{uv}$: total number of edges of the graph.
\end{itemize}
In ERSCD(our method) we are using weighted adjacency matrix,
\begin{equation*}
    A_{uv} = 
\begin{cases}
w_{uv}, & \text{if } u \text{ and } v \text{ are connected with weight w}  \\
0, & \text{otherwise}.
\end{cases}.
\end{equation*}
Clauset's algorithm \cite{clauset2004finding} finds the changes in $\mathcal{Q}$; it considers every possible pair of communities, and for each pair, merges them and calculates how modularity $\mathcal{Q}$ changes, which gives $\Delta \mathcal{Q}_{ij}$ for each possible merge. Merging the communities which increases the $\mathcal{Q}$ most, the above process is repeated until no merge can improve the modularity values.\\
The algorithm begins with every node as a member of community one. The algorithm starts with
\begin{equation*}
        \Delta \mathcal{Q}_{ij} = 
\begin{cases}
\frac{1}{2m} - \frac{d_{i}d_{j}}{(2m)^2}, & \text{ if } $i$ \text{ and } $j$ \text{ are connected}, \\
0, & \text{otherwise}.
\end{cases}
\end{equation*} and 
 \begin{equation*}
     a_{i}=\frac{d_{i}}{2m}
 \end{equation*}
for each $i$.
Their algorithm follows three step process\\
step 1: Initially compute the values of $\Delta \mathcal{Q}_{ij}$ and $a_{i}$ using the above equations. Populate max-heap with the largest value from each row of the $\Delta \mathcal{Q}$ matrix.\\
step 2:
Choose the maximum value of $\Delta \mathcal{Q}_{ij}$ from the max-heap $H$, which contain the largest element of each row of the matrix $\Delta \mathcal{Q}_{ij}$ with the community labels $i$, $j$. Merge the corresponding communities, then update  $\Delta \mathcal{Q}$ matrix,  heap $H$, and value of $a_{i}$ accordingly. Finally, increment $\mathcal{Q}$ by $\Delta \mathcal{Q}_{ij}$.\\
step 3: 
Repeat step 2 until only one community remains.\\
The update rules for $\mathcal{Q}$ are: \\
if a community $k$ is connected to both $i$ and $j$:
\begin{equation*}
    \Delta \mathcal{Q'}_{jk}=\Delta \mathcal{Q}_{ik} + \Delta \mathcal{Q}_{jk}
\end{equation*}
if community $k$ is connected to $i$ but not $j$:
\begin{equation*}
    \Delta \mathcal{Q'}_{jk}=\Delta \mathcal{Q}_{ik} - 2a_{j}a_{k}
\end{equation*}
if community $k$ is connected to $j$ but not $i$:
\begin{equation*}
    \Delta \mathcal{Q'}_{jk}=\Delta \mathcal{Q}_{jk} - 2a_{i}a_{k}
\end{equation*}
Because every merge operation requires $\mathcal{O}((|i| + |j|) \log n)$ time, and the depth of the resultant dendrogram is $d$ i.e., any path in the dendrogram will have at most $d$ nodes. Given the total degree sum is $2m$, the overall runtime is bounded by $\mathcal{O}(m d \log n)$. For further description of the algorithm, one can refer\cite{clauset2004finding}.\\
Now we will discuss the ERSCD algorithm. The algorithm has $3$ steps, effective resistance calculation, followed by sparsification and community detection.\\
\renewcommand{\thealgorithm}{} 



\begin{algorithm}
\caption{ERSCD Algorithm}
\label{alg:erscd}
\textbf{Input:} Graph $G = (V, E)$ and percentage of edge removal $pe$ \\
\textbf{Output:} Detected Community, obtained by:
\begin{enumerate}
  \item Calculate the effective resistance of the input graph $G$
  \item Sparsify the graph for a given $pe$
  \item Give this sparsified graph as input to Clauset’s algorithm for community detection
\end{enumerate}
\end{algorithm}
Let's break down the run time of the proposed algorithm \\
For computing the effective resistance, it will take $\mathcal{O}(1)$ time and computing the pseudo inverse of the Laplacian will take $\mathcal{O}(n^w)$, where $w\leq3$; here pseudo inverse calculation is one time. For sparsification, we are generating the minimum spanning tree, taking $\mathcal{O}(m \log n)$ tme  and run time for community detection by modularity maximization is $\mathcal{O}(m \log^2 n)$. So the total run time of the algo is $\mathcal{O}(n^w + m \log n + m \log^2 n)$. Here $m$ and $n$ denotes graph's edges and node counts, respectively.\\
There exist several algorithms for community detection, of which, some take \textbf{too much computation time} but give \textbf{good accuracy} and others are \textbf{are faster} but  give \textbf{less accuracy} for various datasets. The  present algorithm gives a trade-off and \textbf{balances both time and accuracy},  

Next section we discuss about our experimental results and compare the results with different existing algorithms for community detection.
\section{Experimental Setup}
To validate ERSCD we perform many different experiments in Python. We have done the experiments both in synthetic and real world network datasets. All experiments were conducted on a machine running {Ubuntu 24.04.2 LTS (64-bit)} with an {AMD EPYC\texttrademark~7662} processor.
\paragraph{\textbf{Experimental Data:}}
We use both synthetic and real-world datasets, with synthetic data generated from SBM (Stochastic Block Model) and LFR (Lancichinetti–Fortunato–Radicchi) models.\\
\paragraph{Stocastic Block Model:}
The Stocastic Block Model(SBM)\cite{abbe2018community} is a probabilistic model for generating graph with community structure where $n$ nodes are divided into $k$ separate communities. 
Nodes within the same community are more likely to be connected with  high probability $p_{intra}$. while node-to-node interactions between nodes of different communities occur with lower probability $p_{inter}$. \\
\paragraph{Lancichinetti-Fortunato-Radicch:} The Lancichinetti–Fortunato–Radicchi (LFR)\cite{lancichinetti2008benchmark} benchmark is a popular generative model for assessing community detection algorithms. It produces undirected, unweighted, non-overlapping community graphs whose node degrees and community sizes follow power-law distribution.\\

We are using five real world network data for experiments. First we will describe about the data sets. The data sets are taken from\cite{ni2019community}.
\paragraph{Zachary Karate Club:}
This undirected graph represents club members as nodes, with edges indicating connections between them. It is commonly used in community detection to identify the two group of people.\\
\paragraph{American college football network:}
This undirected network represents football teams as nodes and games between them as edges. The teams are divided into 12 conferences, with intra-conference games occurring more frequently than inter-conference games. 
\paragraph{Political books network:}
The Political Books network is composed of books on US politics for sale on Amazon.com during the 2004 presidential election. Nodes are books, and edges are between the books that were purchased together regularly by the same users.\\
\paragraph{Political blogs network:}
The Political Blogs network was collected around the 2004 U.S. Presidential Election. Nodes represent political blogs, labeled as either liberal or conservative. An edge connects two blogs if one blog cited the other.\\
\paragraph{Cora network:}
This network data is a citation network dataset. Here each node represents a scientific publication and edges denote the citation relationship between papers. The data contains research papers on machine learning, and each paper is classified into a particular topic category. These categories are used as ground truth labels for clustering.\\
Next we will discuss about the Evaluation metric we are using to evaluate the quality of solutions.
\paragraph{\textbf{Evaluation Metric:}}
To evaluate the quality of solutions, we are using various metrics.\\
\paragraph{Adjusted Rand Index (ARI):}
The \textit{Adjusted Rand Index (ARI)}~\cite{santos2009use} evaluates the similarity between two clustering results by adjusting the \textit{Rand Index}, ensuring that random clustering yields a score of 0. \\
Let the ground truth partition of nodes \( \{1, 2, \dots, n\} \) be denoted by disjoint sets \( A_1, A_2, \dots, A_m \), and the predicted clustering by disjoint sets \( B_1, B_2, \dots, B_k \). 
The ARI between the two clustering \( A \) and \( B \) is defined as:
\begin{equation*}
\begin{split}
&\text{ARI}(A, B) =\\ &\frac{
\sum_{i=1}^{m} \sum_{j=1}^{k} \binom{|A_i \cap B_j|}{2} - 
\frac{
\sum_{i=1}^{m} \binom{|A_i|}{2} \sum_{j=1}^{k} \binom{|B_j|}{2}
}{\binom{n}{2}}
}{
\frac{1}{2} \left[
\sum_{i=1}^{m} \binom{|A_i|}{2} + \sum_{j=1}^{k} \binom{|B_j|}{2}
\right] - 
\frac{
\sum_{i=1}^{m} \binom{|A_i|}{2} \sum_{j=1}^{k} \binom{|B_j|}{2}
}{\binom{n}{2}}
}
\end{split}
\end{equation*}
ARI equals 1 when the clustering perfectly matches the ground truth and is less than 1 otherwise.\\
With this we also calculated the Normalized Mutual Information(NMI), which also evaluates the quality of solutions.
\paragraph{Normalized Mutual Information(NMI):}
Let \( \mathcal{A} \) and \( \mathcal{B} \) be two partitions of a network. Define:
\begin{itemize}
    \item \( N_{ij} \): number of nodes in community \( i \) of partition \( \mathcal{A} \) and community \( j \) of partition \( \mathcal{B} \),
    \item \( N_{i\cdot} = \sum_j N_{ij} \): number of nodes in community \( i \) of partition \( \mathcal{A} \),
    \item \( N_{\cdot j} = \sum_i N_{ij} \): number of nodes in community \( j \) of partition \( \mathcal{B} \),
    \item \( n = \sum_{i,j} N_{ij} \): total number of nodes.
\end{itemize}
The Normalized Mutual Information (NMI) is defined as:
\begin{equation*}
\begin{aligned}
\text{NMI}(\mathcal{A}, \mathcal{B}) =
\frac{
-2 \sum\limits_{i=1}^{c_{\mathcal{A}}} \sum\limits_{j=1}^{c_{\mathcal{B}}} N_{ij} \log \left( \frac{N_{ij} \cdot n}{N_{i\cdot} \cdot N_{\cdot j}} \right)
}{
\sum\limits_{i=1}^{c_{\mathcal{A}}} N_{i\cdot} \log \left( \frac{N_{i\cdot}}{n} \right)
+
\sum\limits_{j=1}^{c_{\mathcal{B}}} N_{\cdot j} \log \left( \frac{N_{\cdot j}}{n} \right)}
\end{aligned}
\end{equation*}
if $\mathcal{A} = \mathcal{B}$ then $NMI(\mathcal{A},\mathcal{B})=1$, otherwise less than $1$
\section{Results}
We have experimented with the algorithm on both synthetic and real-world datasets. In both cases, we measure the ARI and NMI values to see the algorithm's effectiveness.  We have also seen that for the $pe$ value between $0$ to $0.5$ (means $0$ percentage to $50$ percentage) we got the optimal result. We picked the $pe$ values with spacing $0.05$.\\
For synthetic data we have used SBM and LFR network data. We have tested our algorithm for different mixing parameter values $\mu$. Here $\mu$ value ranges from $0.1$ to $0.6$. \\
For SBM network data, we have taken $250$ nodes. For LFR network also we have taken $250$ nodes.\\
\begin{table}
    \centering
    \begin{tabular}{cc}
         Number of nodes:  & 250 \\
         Average degree of nodes:  &5 \\
        Power-law exponent for generating the degree distribution:   &3 \\
        Exponent for power law creating community sizes:   &1.5 \\
        Mixing parameter $\mu$: &  [0.1;0.6]\\
      Minimal number of community:  & 10 \\
        Maximal number of community:   & 50 \\
    \end{tabular}
    \caption{LFR-250 parameter setting}
\end{table}

\begin{figure}[H]
    \centering
    \includegraphics[width=0.4\textwidth]{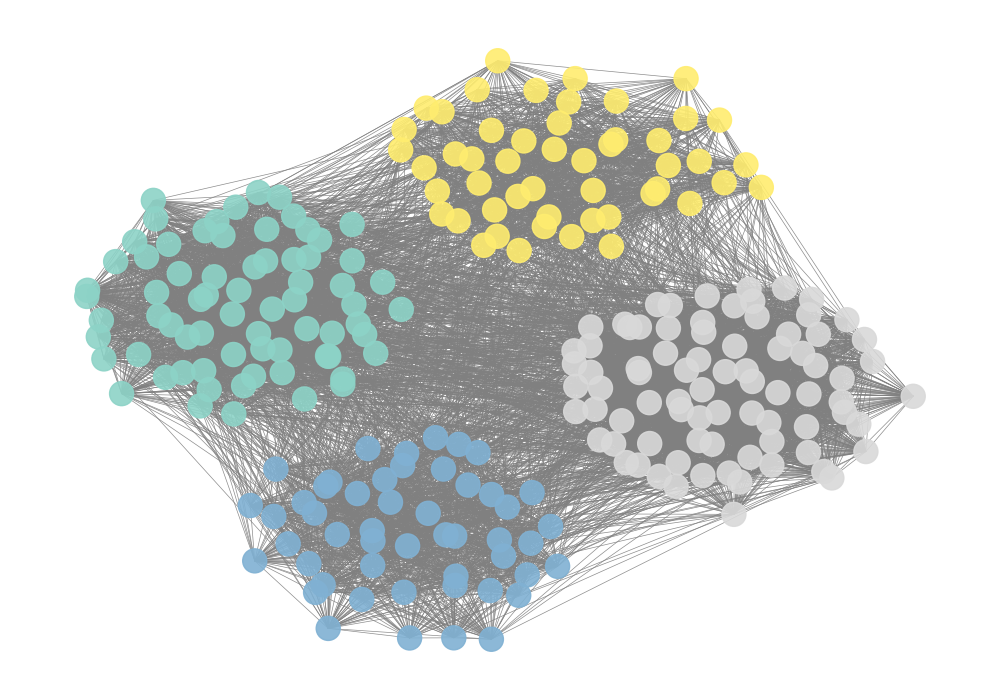}
    \caption{SBM network before community detection}
\end{figure}
\begin{figure}[H]
    \centering
      \includegraphics[width=0.4\textwidth, trim=0 10 0 0, clip]{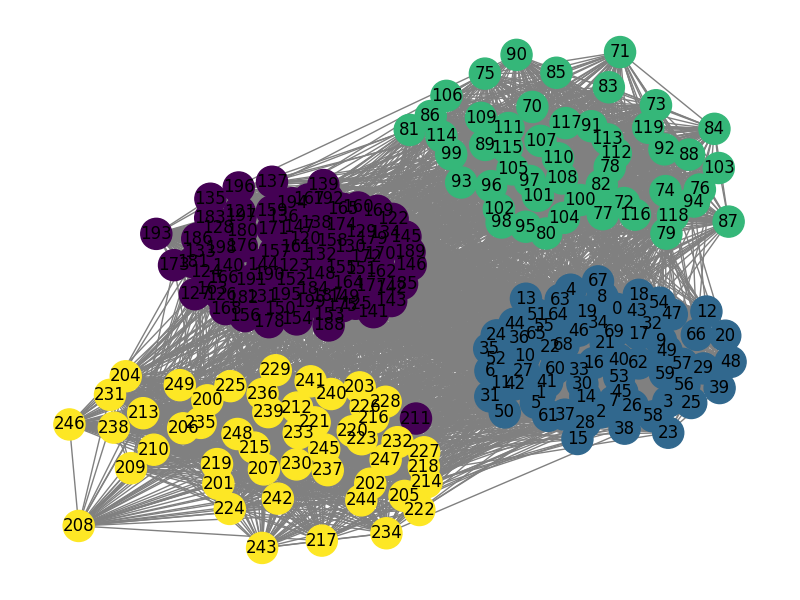}
    \caption{SBM network after community detection}
\end{figure}

For different mixing parameter value the ARI and NMI values will change. When the mixing parameter value is small like $0.1$ or $0.2$, then the ARI value or NMI value does not change much for different values of $pe$ and which is expected. But for the mixing parameter value greater than $0.4$, the ARI and NMI values decrease.\\
 \begin{figure}[H]
    \centering
    \includegraphics[width=0.60\textwidth]{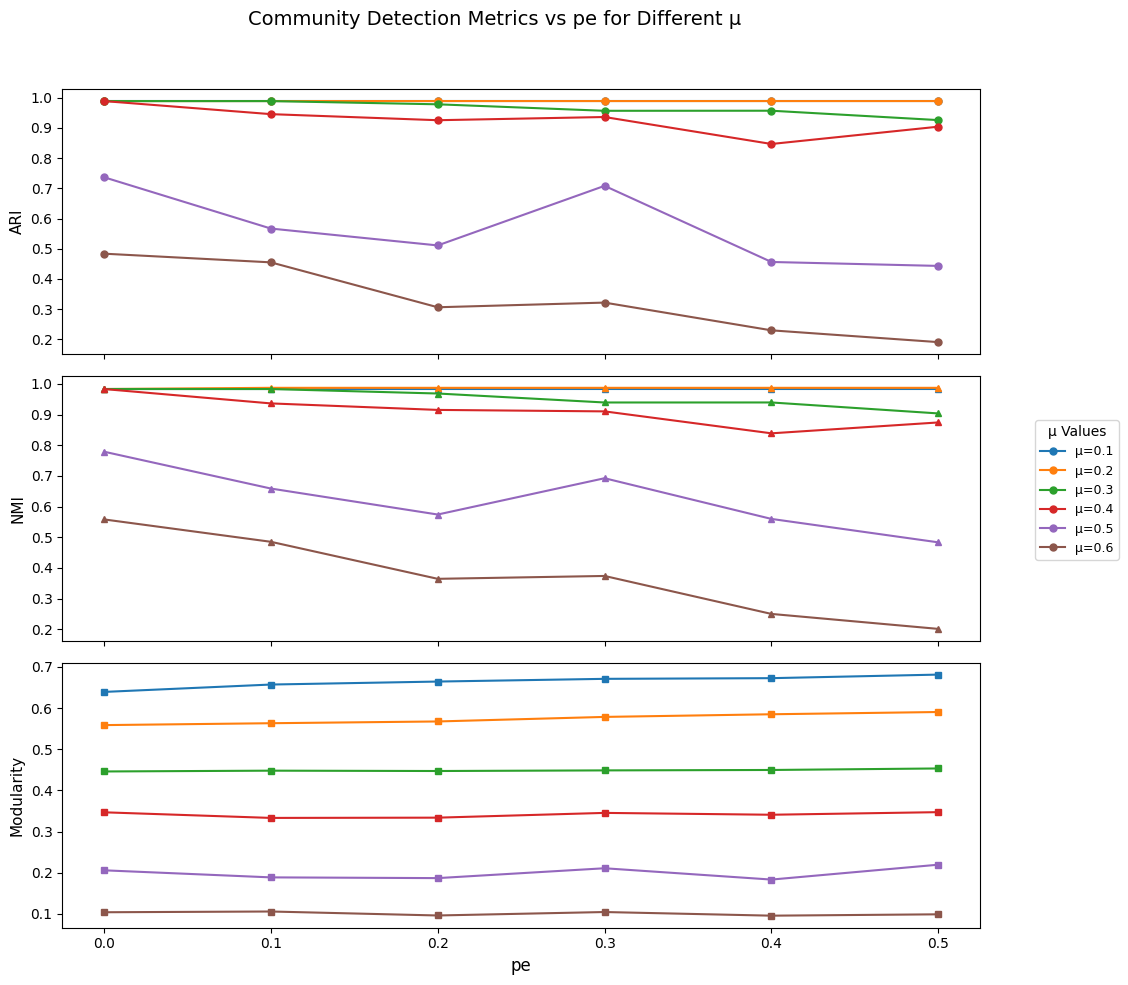}
    \caption{Different metric vs $pe$ for different $\mu$ on SBM}
\end{figure}
 \begin{figure}
    \centering
    \includegraphics[width=0.5\textwidth, trim=0 00 0 0, clip]{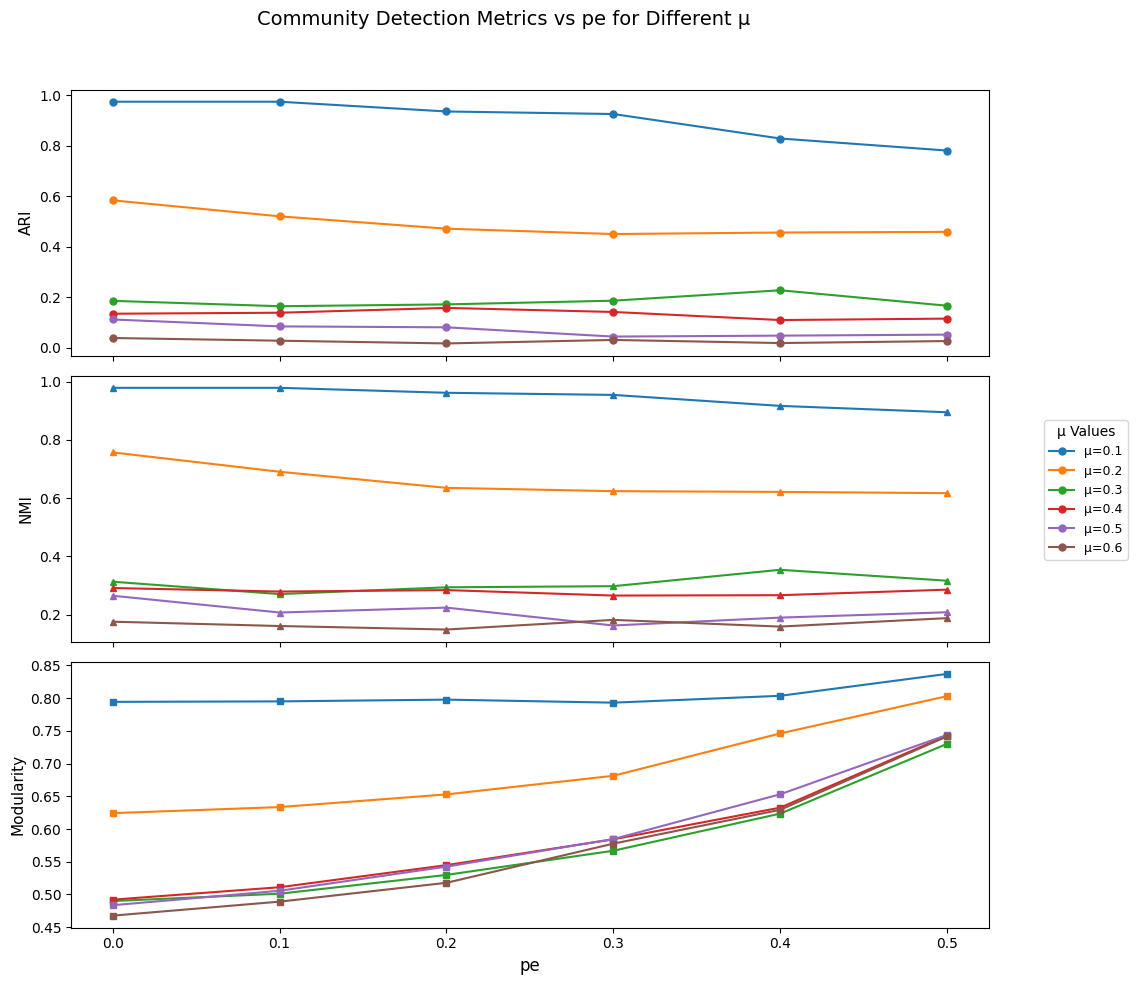}
    \caption{Different metric vs $pe$ for different $\mu$ on LFR}
\end{figure}
For LFR data we got the best NMI and ARI value for small $\mu=0.1$ and $pe=0,0.1$, but for SBM data we got the best NMI and ARI value for $\mu =0.1$ to $0.4$ and for all $pe$ values.\\
It appears that for the synthetic data, sparsification does not play much of a role because the networks are inherently generated with community structure; so the edges are not noisy, and they follow well-defined probabilities. This results in less redundancy to prune.\\
Next, we provide experimental results on different real-world datasets.\\
For the karate club data, we found that the best ARI value is $0.74$ and the NMI value is $0.71$ with modularity $0.42$ for the value $pe$ $0.4$\\
\begin{figure}[H]
    \centering
    \includegraphics[width=0.50\textwidth]{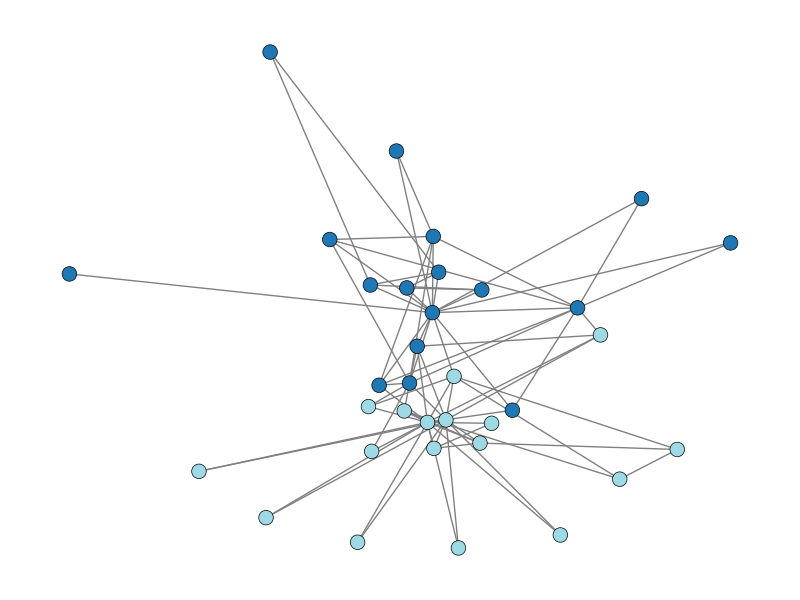}
    \caption{Karate club data}
\end{figure}
\begin{figure}[H]
    \centering
    \includegraphics[width=0.50\textwidth]{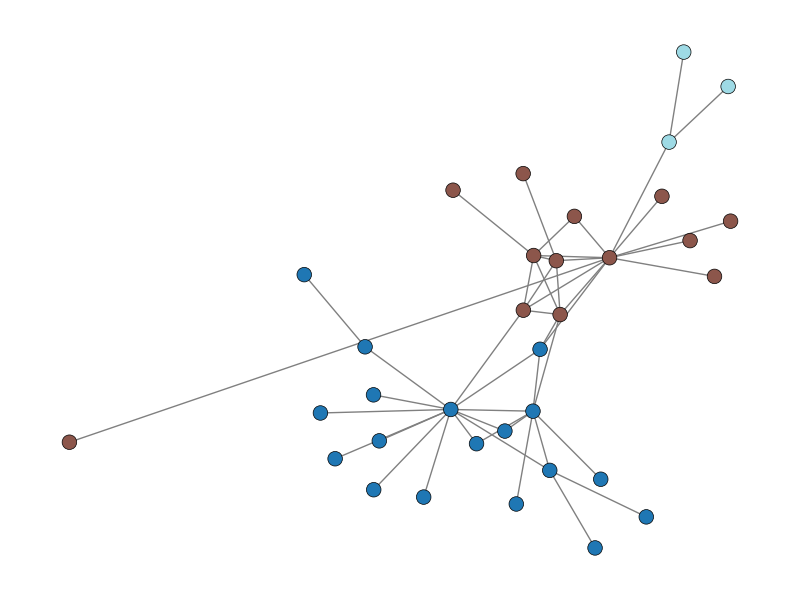}
    \caption{Karate club data after community detection}
\end{figure}
Though the data has $2$ community as ground truth but after algorithm runs, it finds the two community, and only $3$ data points are outside the community. Next we will show how the ARI, NMI and Modularity value changes for karate club data for different values of $pe$.\\
\begin{figure}[H]
    \centering
    \includegraphics[width=0.70\textwidth]{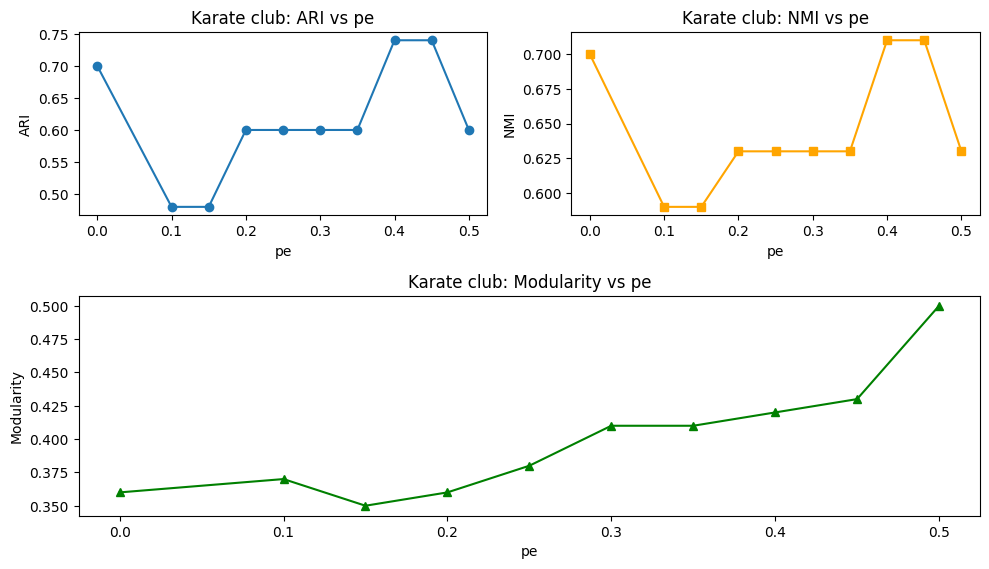}
    \caption{Changes of ARI,NMI and Modularity for different pe value}
\end{figure}
For American college football network, we got the best ARI $0.87$ and NMI $0.91$ with modularity $0.87$ for $pe$ value $0.35$.\\
 \begin{figure}[H]
    \centering
    \includegraphics[width=0.5\textwidth]{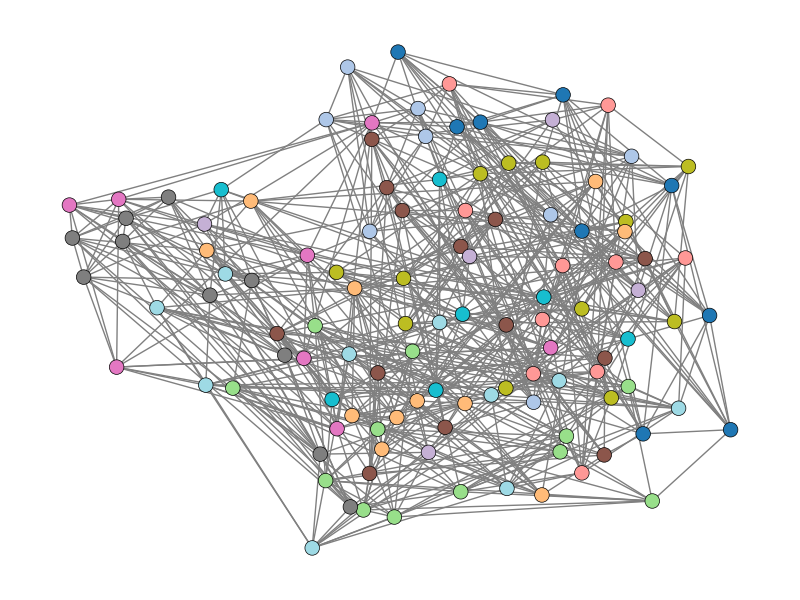}
    \caption{American college football data}
\end{figure}
 \begin{figure}[H]
    \centering
    \includegraphics[width=0.5\textwidth]{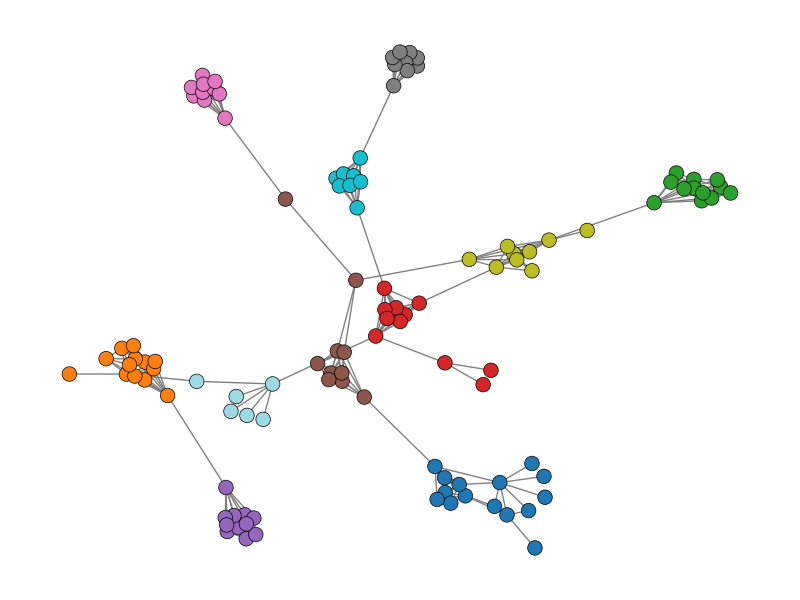}
    \caption{American College football data after community detection}
\end{figure}
The actual data has $12$ communities and it detects all the communities properly. Next plot shows how the ARI and NMI, and Modularity
value changes for different values of $pe$.\\
 \begin{figure}[H]
    \centering
    \includegraphics[width=0.50\textwidth]{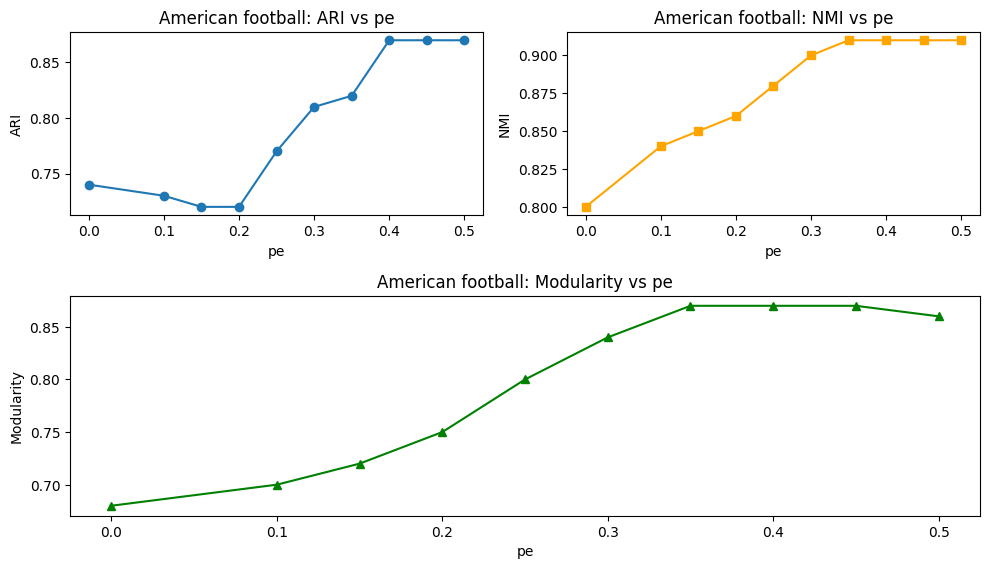}
    \caption{Changes of ARI,NMI and Modularity for different pe value}
\end{figure}
For political books data set, we got the optimal ARI $0.70$ and NMI $0.59$ with modularity $0.51$ for $pe$ value $0.2$. The data has $3$ ground truth communities and the algorithm detects all the communities properly.\\
 \begin{figure}[H]
    \centering
    \includegraphics[width=0.35\textwidth]{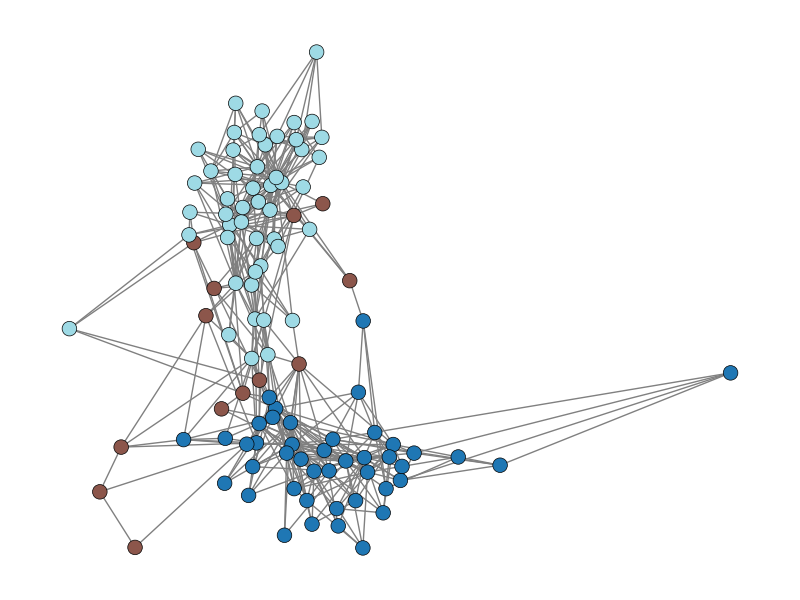}
    \caption{Political Books data}
\end{figure}
 \begin{figure}[H]
    \centering
    \includegraphics[width=0.35\textwidth]{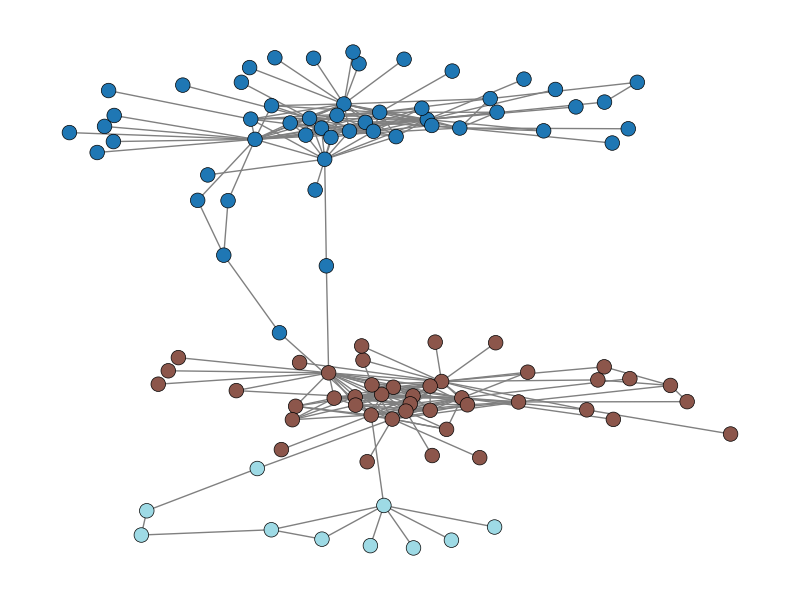}
    \caption{Political Books data after community detection}
\end{figure}
The next plot will demonstrate the changes of ARI ,NMI, and Modularity value for different $pe$ values.\\
 \begin{figure}[H]
    \centering
    \includegraphics[width=0.50\textwidth]{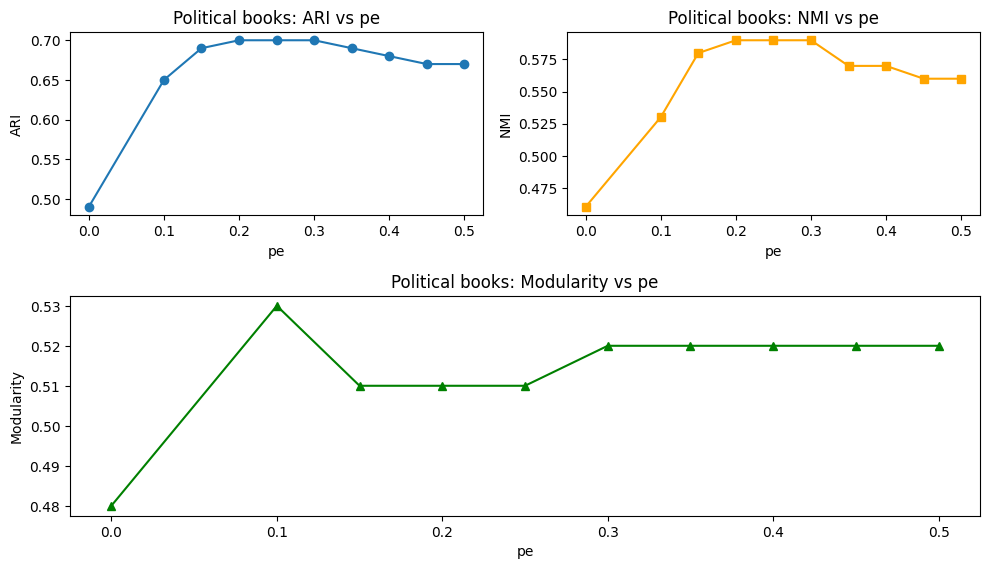}
    \caption{Changes of ARI,NMI and Modularity for different pe value}
\end{figure}
For political blog data, we got the best ARI $0.80$, NMI $0.67$, and the modularity value $0.42$ for the $pe$ value $0$, which means here we are not removing any edges. The data has two communities as ground truth, and it detects the communities. The plot shows how the ARI, NMI and modularity values are changing for different $pe$ values.\\
 \begin{figure}[H]
    \centering
    \includegraphics[width=0.50\textwidth]{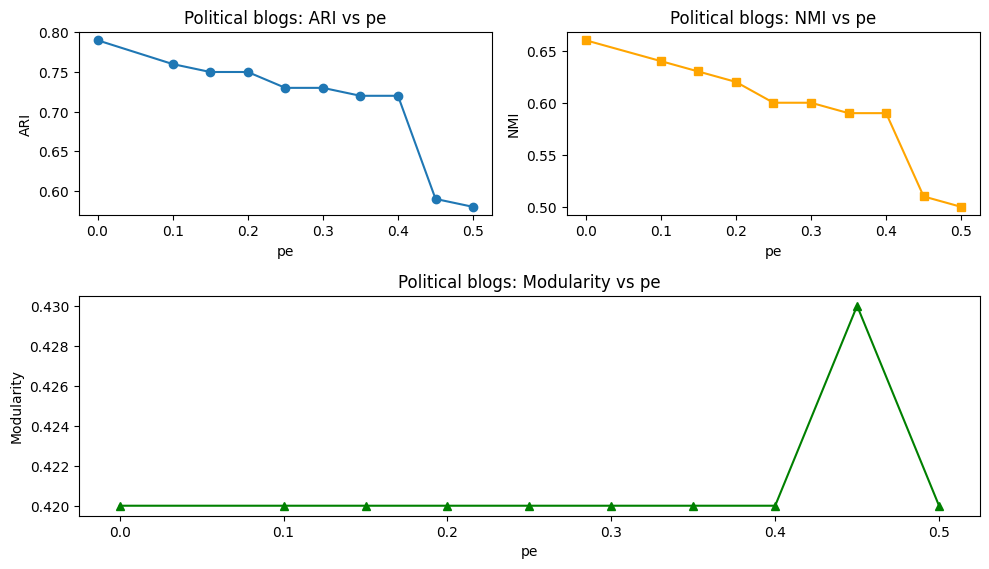}
    \caption{Changes of ARI,NMI and Modularity for different pe value}
\end{figure}
For Cora networkdata, we got the best ARI $0.31$, NMI $0.46$, and the modularity value $0.90$ for the $pe$ value $0$, which means here we are not removing any edges. The data has seven communities as ground truth, and it detects the communities with low accuracy. Here the ARI value is not that much good but the modularity value is high, which means the network has a strong community structure.  The plot shows how the ARI, NMI and modularity values are changing for different $pe$ values.\\
 \begin{figure}[H]
    \centering
    \includegraphics[width=0.50\textwidth, trim=0 150 0 0, clip]{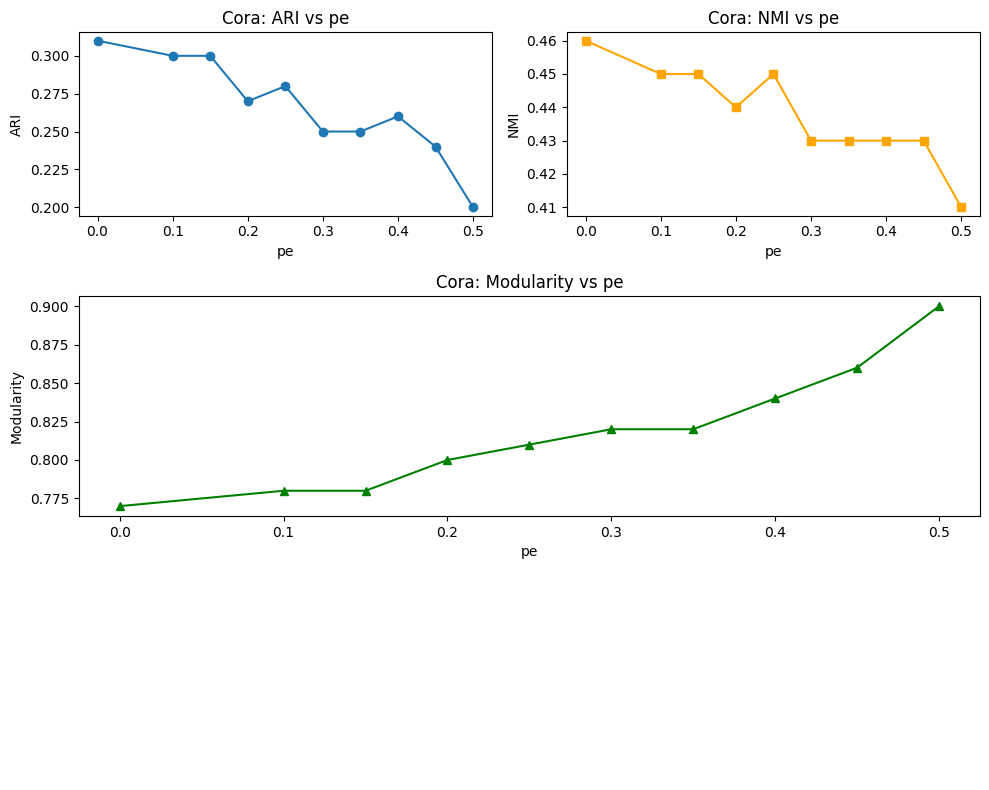}
    \caption{Changes of ARI,NMI and Modularity for different pe value}
\end{figure}
\section{Evaluation Against Existing Methods}
We are comparing our method(ERSCD) with three benchmark methods. These are Louvain\cite{blondel2008fast}, Infomap\cite{rosvall2008maps} and Ricci flow\cite{ni2019community}.\\
\paragraph{\textbf{Louvain:}}
The Louvain method is a very efficient community detection algorithm widely used on large networks. It follows the optimization of modularity, which is a quality function assessing the density of links within communities relative to links between communities. The approach proceeds in a greedy two-step manner: it first assigns each node to its own community, then moves nodes between communities in order to maximize modularity gain locally. In the second stage, communities are merged to constitute a new network, and iteration is carried out until no modularity increase can be achieved any further. 
\paragraph{\textbf{Infomap:}}
The Infomap algorithm is based on information theory and treats community detection as a problem of finding the simplest way to describe the flow of information on a network. The idea is to minimize how much information is needed to describe the movement of a random walker through the graph. Infomap works by grouping nodes that are frequently visited by the walker, helping to find communities within the network. 
\paragraph{\textbf{Ricci Flow:}}
The Ricci Flow method is inspired by the concept of Ricci curvature from geometry and is used for community detection in networks. It works by treating the network as a space that evolves over time, with the goal of smoothing out uneven connections between nodes. The method adjusts the "curvature" of the network, aiming group nodes into communities where the connections within each group are stronger than those between groups. The Ricci flow process helps identify natural clusters in the network by progressively balancing the graph structure.\\

We compared all the algorithms on both synthetic and real world network. For synthetic data we have used SBM and LFR with $250$ nodes and different mixing parameters. The measurements for Louvain and Infomap are taken as the avarage of 10 simulation for synthetic network. \\
 \begin{figure}[H]
    \centering
    \includegraphics[width=0.42\textwidth]{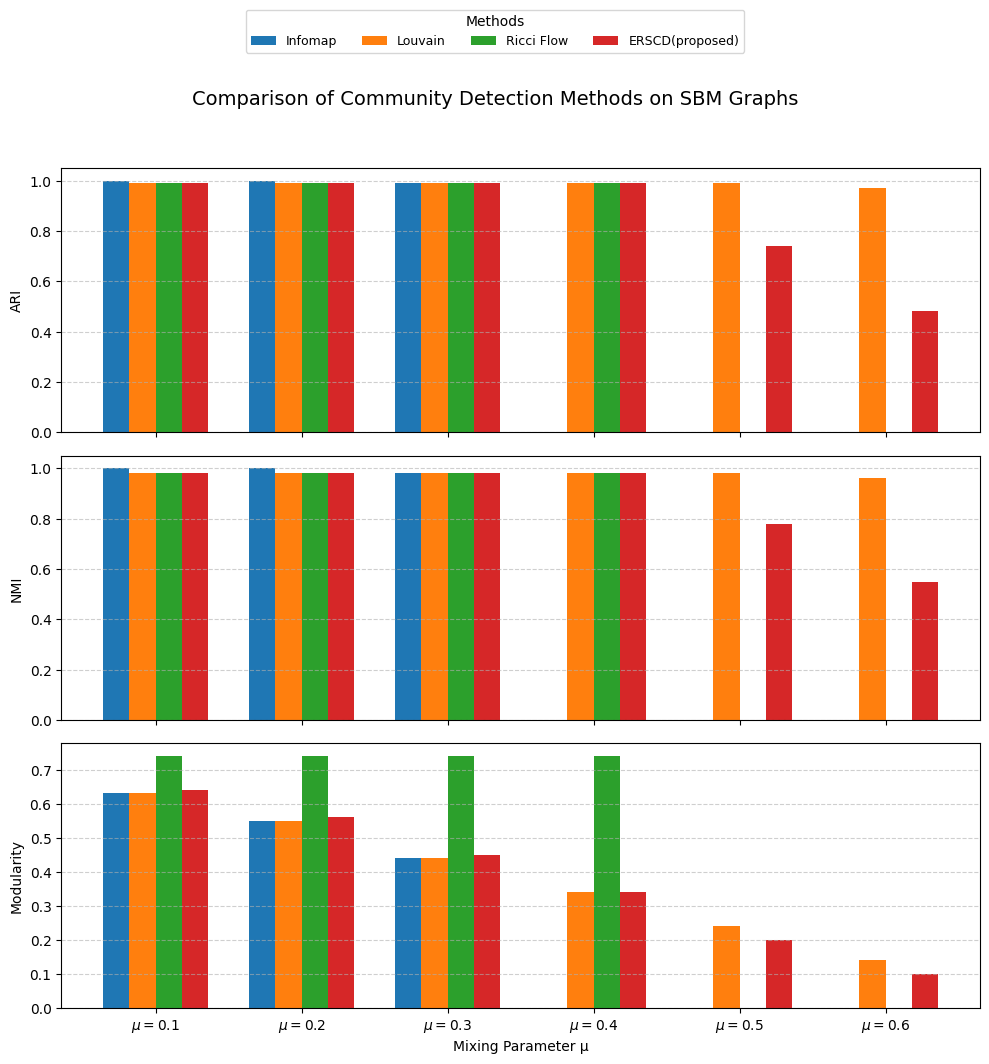}
    \caption{comparison of different methods for differen $\mu$ values}
\end{figure}
 \begin{figure}[H]
    \centering
    \includegraphics[width=0.42\textwidth]{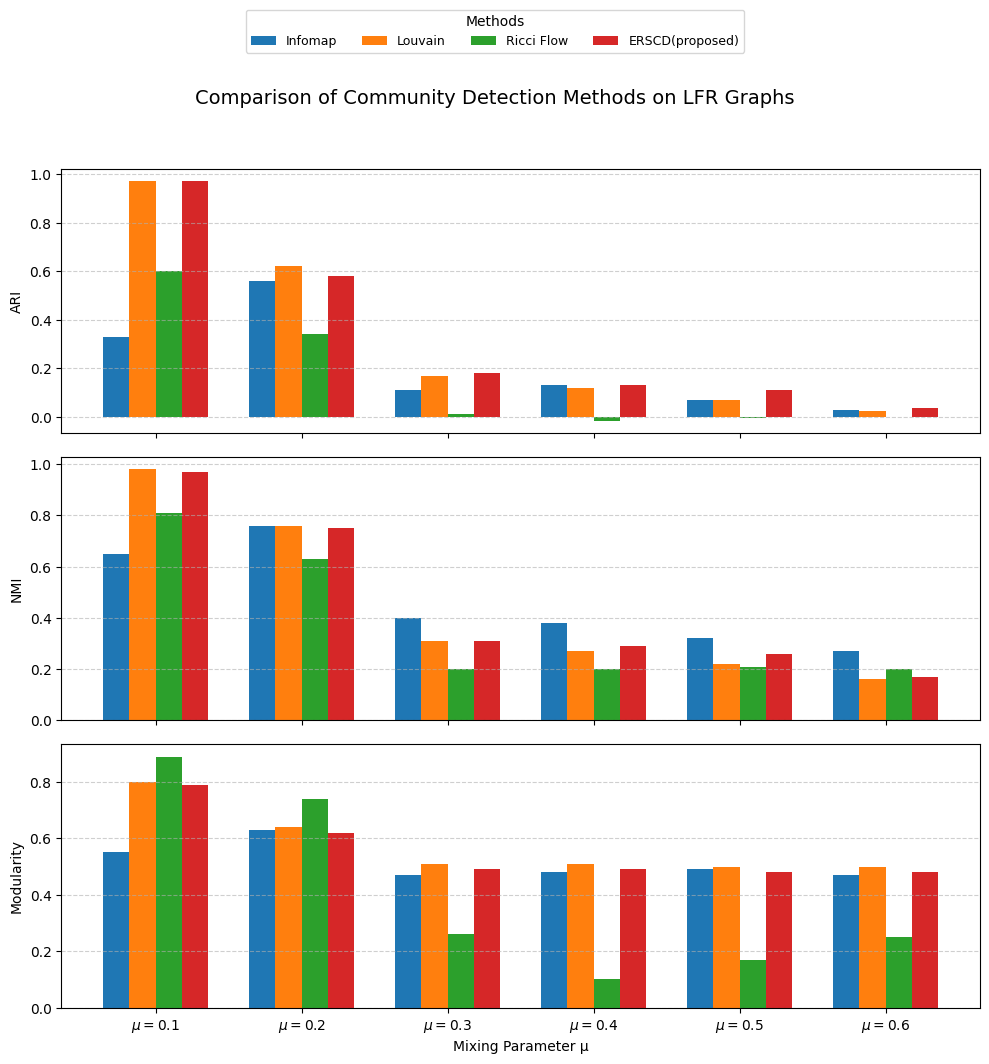}
    \caption{comparison of different methods for differen $\mu$ values}
\end{figure}
It is seen that in case of SBM data, for less $\mu$ value all methods perform outstandingly, but for $\mu$ greater than $0.4$, only two methods are working, one is our method, other is the Louvain method. In paper\cite{ni2019community} also they mentioned that the Ricci flow method will work for $\mu \leq 0.4$. For LFR data also same thing happens but here no method can give better ARI and NMI value for $\mu=0.2$. In some cases the Ricci flow method gives negative ARI values too.\\
For different real-world network data, we compute the ARI value of these three methods, including the proposed method ERSCD, it is seen that ERSCD \textbf{outperforms} all the other methods in most of the cases.\\
 \begin{figure}[H]
    \centering
    \includegraphics[width=0.40\textwidth]{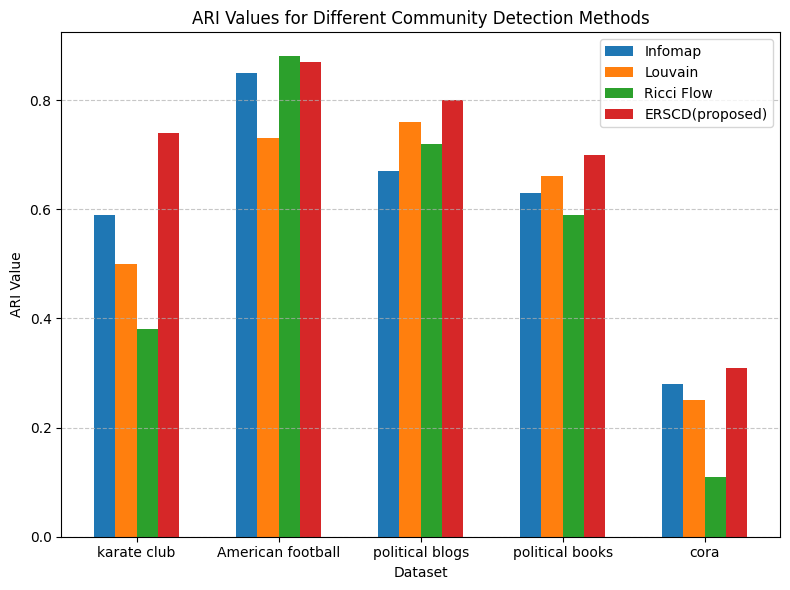}
    \caption{ARI values of different methods for different datasets}
\end{figure}
We have seen the same trends for the NMI values. Here, also in most cases ERSCD \textbf{performs} better.\\
 \begin{figure}[H]
 \vspace{-1mm}
    \centering
    \includegraphics[width=0.40\textwidth]{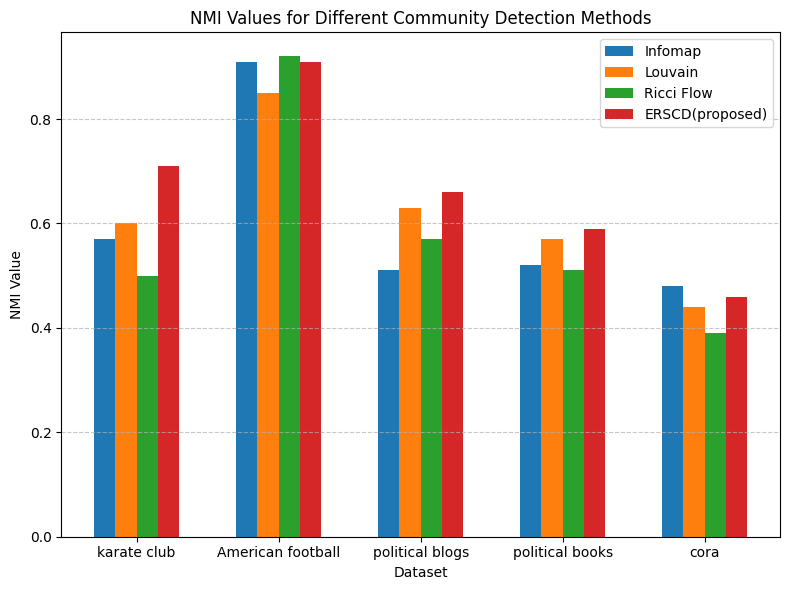}
    \caption{NMI values of different methods for different datasets}
\end{figure}
In the case of modularity value, we have seen that the Ricci flow method's performance is better than the others.\\
 \begin{figure}[H]
    \centering
    \vspace{-5mm}
    \includegraphics[width=0.40\textwidth]{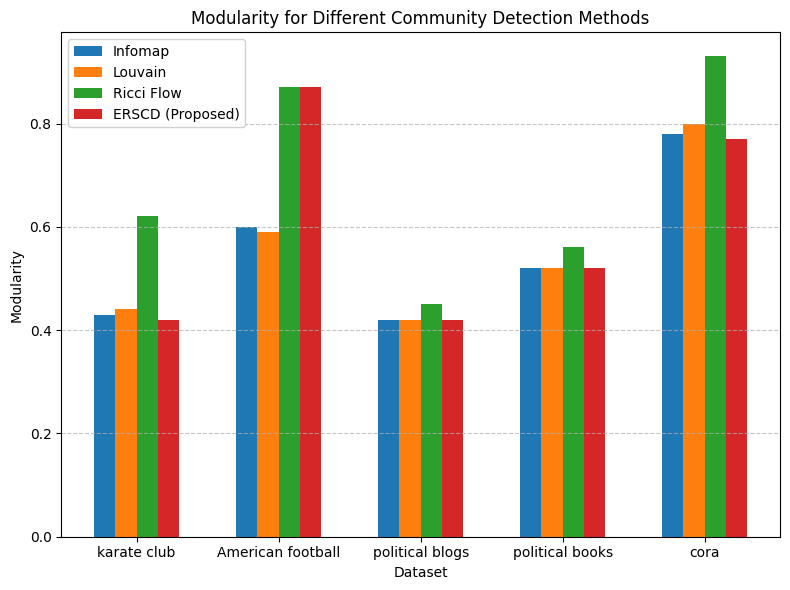}
    \caption{Modularity values of different methods for different datasets}
    \vspace{-5mm}
\end{figure}
We have compared the execution time of different methods on different data sets. Result shows that Ricci flow method is taking the \textbf{highest time}, and \textbf{least time} is taken by Infomap and Louvain method, and our method's execution time lies in \textbf{between}. It means our algorithm takes \textbf{reasonable} time and gives \textbf{better} ARI value than the other methods.\\ The plot below is done in ${log}$ scale.\\
 \begin{figure}[H]
    \centering
    \includegraphics[width=0.40
    \textwidth]{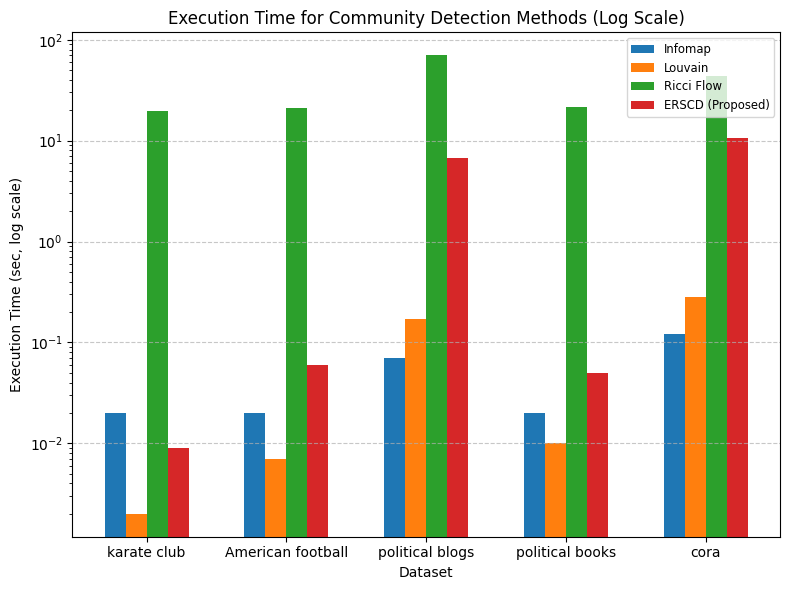}
    \caption{Execution time of different methods for different datasets}
\end{figure}

\section{Conclusion and Discussion}
We have introduced ERSCD, a new approach for detecting communities in an undirected connected graph. In this method, we model the graph as an electrical circuit and compute the effective resistance between each pair of connected nodes. This resistance-based distance is then used to reweight the graph and delete a percentage of edges from non-MST edges. It guides the search for communities that maximize weighted modularity.\\
Through extensive testing on both synthetic and real-world datasets, our result shows that in most cases ERSCD \textbf{consistently performs better} in terms of accuracy and also takes \textbf{reasonably less computing time} than the other existing methods .\\
Another advantage of our algorithm is that it \textbf{does not  involves  many parameters} like other methods. Only one parameter of our algorithm is $pe$, which decides how many edges to remove from the network. By checking the different $pe$ value experimentally from $0$ percent to $50$ percent, we got a good result. However, further work is needed to define a more general rule for selecting the $pe$ value.\\
We haven't used GPU acceleration yet, but doing so could greatly speed up the runtime.
\section{Funding}
Research of JP is supported by JRF from Department of Science and Technology(DST), Government of India, Research of PB is funded by JRF from the Ministry of Human Resource Development (MHRD), Government of India.

\bibliographystyle{siamplain}
\bibliography{bibliography}

\end{document}